\documentclass[10pt,journal]{IEEEtran}

\ifCLASSINFOpdf
\else
\fi

\usepackage{epstopdf}
\usepackage{graphicx}
\usepackage{color}
\usepackage{soul}
\usepackage{graphics}
\usepackage{amsfonts}
\usepackage{amssymb}
\usepackage[cmex10]{amsmath}
\usepackage{array}
\usepackage{comment}
\usepackage{booktabs}
\usepackage{subfig}
\usepackage{listings}
\usepackage{url}
\usepackage{mdframed}
\usepackage{epstopdf}
 \usepackage{gensymb}
\usepackage{amsthm}
\usepackage{tikz}
\usepackage[linesnumbered,ruled,vlined]{algorithm2e}

\newcommand{\ignore}[1]{}

\newtheorem{lemma}{Lemma}

\newtheorem{theorem}{Theorem}

\usepackage{mathtools}

\usetikzlibrary {positioning}
\usepackage{algorithm2e}

\begin{document}
%!TEX TS-program = latex

\title{1D Modeling of Sensor Selection Problem for Weak Barrier Coverage and Gap Mending in\\ Wireless Sensor Networks}
%\title{Optimizing Quality of View and Resource Utilization in Visual Sensor Networks}

%\author{
%  \IEEEauthorblockN{S.M.Reza Soroushmehr\IEEEauthorrefmark{1}, Hamed Sadeghi\IEEEauthorrefmark{2}, Shahrokh Valaee\IEEEauthorrefmark{2}, Shahram Shirani\IEEEauthorrefmark{3},  Shadrokh Samavi\IEEEauthorrefmark{3}}
%   \\\IEEEauthorblockA{\IEEEauthorrefmark{1}University of Michigan
%  \\ssoroush@umich.edu}
%     \\\IEEEauthorblockA{\IEEEauthorrefmark{2}ECE Department, University of Toronto
%   \\hamed.sadeghi@mail.utoronto.ca, valaee@comm.utoronto.ca}
%    \\\IEEEauthorblockA{\IEEEauthorrefmark{1}ECE Department, McMaster University
%  \\\{ samavi, shirani\}@mcmaster.ca}
%}

\author{Hamed~Sadeghi,~\IEEEmembership{Student~Member,~IEEE,}
MohammadReza Soroushmehr,~\IEEEmembership{SPS~Member,~IEEE,}
        Shahrokh~Valaee,~\IEEEmembership{Senior~Member,~IEEE,}
        Shahram~Shirani,~\IEEEmembership{Senior~Member,~IEEE}
        and~Shadrokh~Samavi,~\IEEEmembership{Senior~Member,~IEEE}% <-this % stops a space

%\IEEEcompsocitemizethanks{\IEEEcompsocthanksitem H. Sadeghi and S. Valaee are with the Department
%of Electrical and Computer Engineering, University of Toronto, ON, Canada, M5S 2E4.\protect\\
%% note need leading \protect in front of \\ to get a newline within \thanks as
%% \\ is fragile and will error, could use \hfil\break instead.
%E-mail: hsadeghi,valaee@comm.utoronto.ca.}

\thanks{H. Sadeghi and S. Valaee are with the department
of Electrical and Computer Engineering, University of Toronto, ON, Canada, M5S 2E4.\protect\\
E-mail: hsadeghi,valaee@ece.utoronto.ca}

%\IEEEcompsocthanksitem{ S. Shirani is with the Department of Electrical and Computer Engineering, McMaster University, ON, Canada, L8S 4L8.
%Email: shirani@mcmaster.ca.}% <-this % stops an unwanted space

\thanks{M.R. Soroushmehr is with the MCIRCC center, University of Michigan, MI, 48109.
Email: ssoroush@umich.edu}

\thanks{S. Shirani is with the department of Electrical and Computer Engineering, McMaster University, ON, Canada, L8S 4L8.
Email: shirani@ece.mcmaster.ca}

\thanks{S. Samavi is with the department of Electrical and Computer Engineering, McMaster University, ON, Canada, L8S 4L8.
Email: samavi@mcmaster.ca}}

%\markboth{IEEE Transactions on Mobile Computing}%
%{Shell \MakeLowercase{\textit{et al.}}: Bare Demo of IEEEtran.cls for Computer Society Journals}

\maketitle

\begin{abstract}
In this paper, we first remodel the line coverage as a 1D discrete problem with co-linear targets. Then, an order-based greedy algorithm, called OGA, is proposed to solve the problem optimally. It will be shown that the existing order in the 1D modeling, and especially the resulted Markov property of the selected sensors can help design greedy algorithms such as OGA. These algorithms demonstrate optimal/efficient performance and have lower complexity compared to the state-of-the-art. Furthermore, it is demonstrated that the conventional continuous line coverage problem can be converted to an equivalent discrete problem and solved optimally by OGA.
Next, we formulate the well-known weak barrier coverage problem as an instance of the continuous line coverage problem  (i.e. a $1$D problem) as opposed to the conventional 2D graph-based models. We demonstrate that the equivalent discrete version of this problem can be solved optimally and faster than the state-of-the-art methods using an extended version of OGA, called $K$-OGA. 
Moreover, an efficient local algorithm, called LOGM, is proposed to mend barrier gaps due to sensor failure. In the case of $m$ gaps, LOGM is proved to select at most $2m-1$ sensors more than the optimal while being local and implementable in distributed fashion.
% Moreover, an optimal order-based method, OGMMS, is proposed for minimizing the movement of mobile sensors used to mend the gaps. We prove that OGMMS is less complex compared to the state-of-the-art.
%Finally, we derive a condition to verify the convertability of the path coverage problem and propose an efficient algorithm to check it.
We demonstrate the optimal/efficient performance of the proposed algorithms via extensive simulations.
\end{abstract}


\IEEEpeerreviewmaketitle

%%%%%%%%%%%%%%%%%%%%%%%%%%%%%%%%%%%
\section{Introduction}\label{sec:introduction}

Coverage is a fundamental problem in wireless sensor networks \cite{cardei04, hou09} that has been defined in different forms including point coverage, path coverage, line coverage and  barrier coverage. 

%types of coverage
In the point coverage problem, a number of targets in the area of interest are desired to be $K$-covered, i.e. covered by at least $K$ sensors \cite{fusco09, fan08, wang07}. 
%In area coverage, an area of interest is desired to be covered considering teh coverage area of sensors \cite{liu11}.
In line coverage, a number of target lines are to be covered by the available sensors \cite{ prakash14, dash12, barun14}. Applications of line coverage include monitoring a hallway, a road network or a border line for intrusion detection. 
Path coverage ensures the trackability of moving objects traversing a specified path in the environment \cite{harada08, ram07}. As an example, for the localizability of the moving objects, any point on the trajectory should be covered by a number of sensors \cite{ram07}. 

%2D modelling
Coverage problem has been modeled as a 2D problem in the majority of the literature works.
In fact, the common approach is to model the sensor coverage area as a 2D region and consider 2D points (point coverage) \cite{fusco09}, 2D areas (area coverage) \cite{he15, sharmin15}, 2D overlap among sensors (strong/weak barrier coverage) \cite{kumar05, wang14} or 2D trajectories (path coverage) \cite{ram07} depending on the application.  Even for the line coverage problem, the majority of works use 2D models and  consider lines as 2D targets \cite{prakash14, dash12} and propose to cover them with 2D coverage areas of sensors.

%barrier coverage
Barrier coverage is a special type of coverage coined by \cite{kumar05}. In barrier coverage, we would like to detect any intruder crossing a belt-shaped region at least by $K$ sensors, i.e. $K$-barrier coverage. 
The problem is categorized into strong and weak $K$-barrier coverage.
%There exist two types of barrier coverage: strong barrier coverage and weak barrier coverage \cite{kumar05}.
In the strong $K$-barrier coverage, researchers deploy/select/modify/move sensors to ensure that if an intruder takes `any' paths crossing the region, it will be detected  by at least $K$ sensors. Weak barrier coverage is a less restricted type in which only parallel paths (perpendicular to the region boundaries) are guaranteed to be covered at least $K$ times. 
Barrier coverage has applications in border protection \cite{wang14},  building surveillance \cite{kumar05}, etc.

%K-coverage
To study $K$-barrier coverage, a number of works have focused on sensor deployment \cite{gong16, he14}. 
Others have focused on after deployment sensor modification (selection/rotation/scheduling) to ensure $K$-coverage \cite{wu15, chen13, chen13_2}.
Furthermore, mobile sensors have been proposed in a number of works \cite{wang14, du13, keung12} to mend the gaps arising due to sensor failure, unsuccessful deployment/modification attempts, etc.

%graph-based modeling
In the majority of barrier coverage works, the overlap of the sensors has been modeled using a graph called coverage graph \cite{kumar05, wang14}. In the coverage graph, nodes and edges represent the sensors and their overlap, respectively. In strong barrier coverage graphs, edges show actual overlap while in weak coverage graph they represent projected overlap \cite{wang14}. Projected overlap is the overlap existing among the 1D intervals resulted from the projection of 2D coverage areas onto the boundary of the region.
%In fact, the majority of works have focused on strong barrier coverage to be able to detect intruders with arbitrary paths and address the weak barrier coverage as a special case. 
%In this paper, we focus on weak barrier coverage and path coverage problems as well as their related issues, for example mending gaps (due to sensor failure, etc) by sensor selection/rotation \cite{chen13} or using mobile sensors \cite{wang14}. 
%Some of these works have considered 2D Poisson (i.e. effectively uniform) \cite{} distribution of sensors and other used line-based deployment to design their modification strategy.
%There has also been some works on 3D modeling of the coverage to make the considered scenario more realistic \cite{}.

 % In this paper
 In this paper, we study the sensor selection problem for the weak $K$-barrier coverage and weak barrier gap mending.
 Although strong barrier coverage is a 2D problem, the weak barrier coverage is essentially a 1D problem due to the presence of order. Hence, we exploit the existing order and convert the problem to an equivalent discrete line coverage ones defined later. The equivalent problem is solved optimally by a novel greedy algorithm called OGA. OGA is proved to be optimal in terms of selecting the minimum number of sensors and linearly complex in the number of sensors. The Markov property of the selected sensors defined later is a key point in proving the optimality of OGA.
Furthermore, we study the weak barrier gap mending problem and demonstrate that the existing order in this 1D problem enables us to propose a local algorithm similar to OGA, called LOGM. LOGM is efficient and can be implemented in a distributed form. 

 %Here, we first redefine the line coverage problem and recognize two versions of it, i.e. continuous and discrete line coverage. Then,  we propose an algorithm, called OGA, to solve both versions of the line problem optimally. OGA  is a greedy algorithm that exploits the order existing among the points located on a line. We prove that OGA is optimal in terms of selecting the minimum number of sensors and its complexity is linear in the number of sensors.
%Next, we consider some practical issues in the barrier coverage application and demonstrate how we can exploit the line coverage problem, the OGA method, and the existing order in the 1D problem to propose algorithms with optimal performance and/or reduced complexity. 
  %In the end, we study the path coverage problem and derive a condition to check if the path coverage problem is convertible to an equivalent line coverage one. We call such problems `linearable'.  Furthermore, we propose a fast algorithm  to verify the derived condition in reasonable time. 
 
%contributions
The main contributions of our work are three folds: First, we introduce a novel (discrete) formulation of the line coverage problem and prove that the continuous case is a special case of it. Then, we propose an optimal greedy algorithm, called OGA, which is proved to be linear-time.
Second, we model the weak barrier coverage as a line coverage problem and propose a generalized version of the OGA problem, called $K$-OGA, to solve it optimally.
Third, we propose a local algorithm, called LOGM, to select sensors efficiently to mend the gap due to sensor failure. The proposed algorithm is proved to select at most $2m-1$ sensors more than the optimal in the $m$-gap scenario.
%Fourth, we propose an optimal method, called OGMMS, for using mobile sensors to mend the gaps due to sensor failure. We show that the proposed algorithm is less complex compared to the state-of-the-art. 
%In the end, we derive a condition to verify linearability and propose a reasonable time algorithm to perform the check.

The remaining sections are organized as follows. In Section \ref{sec:relatedwork}, we review a number of related works and explain the ones used as benchmarks in detail. Section \ref{sec:mathmodel} provides the mathematical models and novel algorithms proposed to solve them.
Section \ref{sec:kbarrier} and \ref{sec:failure} study the $K$-barrier coverage and the gap mending as instances of the line coverage problem.
Section \ref{sec:simres} demonstrates the better performance of the proposed algorithms numerically. Finally, Section \ref{sec:conclusion} concludes the paper.

%%%%%%%%%%%%%%%%%%%%%%%%%%%%%
\section{Related work}\label{sec:relatedwork} 

Barrier coverage and gap mending have been extensively studied in the past decade. Here, we review some recent state-of-the-art works from each application.

% Less important works
Fusco and Gupta \cite{fusco09} proposed a greedy algorithm for selection and orientation of sensors to achieve simple $K$-coverage of targets, i.e. point coverage. We design a greedy benchmark in our simulations inspired by their algorithms.

% line coverage
A number of algorithms have been proposed in the literature to address the line coverage problem \cite{prakash14, dash12}. 
The problem of line coverage for intrusion detection is studied in \cite{prakash14}. It is assumed that sensors are dropped along a line with random deviations and the line coverage and connectivity of the sensors are studied.

The problem of $1$-covering of a number of lines by using the minimum number of sensors is studied in \cite{dash12}. It is proved that the problem is NP-hard and  two constant-factor and one polynomial-time approximation algorithms are proposed to solve the problem for axis-parallel line segments.  

%barrier coverage
Kumar \emph{et. al.} coined the term `barrier coverage' for the first time \cite{kumar05}. They distinguish strong and weak barrier coverage and propose efficient algorithms to determine whether $K$ barriers exist after the deployment. Furthermore, they propose optimal deterministic deployment to achieve barrier coverage. Finally, they study strong and weak barrier coverage probabilistically for a random deployment of sensors.

%Some works have studied  the probability of existence of $K$-barrier under guided random deployment of sensors \cite{prakash14}. 

A survey on barrier coverage is performed by Tao and Wu \cite{tao15}. They mention that heterogeneity and low complexity are two less addressed problems in the barrier coverage in directional sensor networks. In concord with these demands, the formulations and algorithms proposed in our paper are independent of the sensing (coverage) models and are less complex compared to the state-of-the-art. 

%Benchmarks
% Main benchmark
A thorough study of $K$-barrier coverage using both static and mobile directional sensors has been performed by Wang et. al. \cite{wang14}.
First, they investigate the minimum number of sensors required to form strong/weak $K$-barriers. Second, they propose a method to minimize the movement (i.e. energy dissipation) of mobile sensors required for gap mending. To solve the first problem, they propose the concept of weighted barrier graph, which shows the so called strong/weak `overlap' among sensors. Next, they prove that the existence of $K$ barriers is equivalent to having $K$ vertex-disjoint paths in a transformed version of the proposed weighted graph. More importantly, they propose a greedy algorithm that selects $K$ vertex-disjoint paths, one after another with removal of the selected sensors, by using the Dijkstra algorithm. We use this algorithm as one of the main benchmarks in our simulations.
 %Second, they propose to use Hungarian method to assign mobile sensors to gaps to mend them with minimum required movement, i.e. minimum energy.

% gap mending
Gap mending for barrier coverage is another important application in the wireless sensor networks. The majority of works have proposed to rotate sensors to mend the barrier gaps \cite{wu15, chen13}. A centralized chain-reaction-based method has been proposed in \cite{chen13} to heuristically mend the gaps by starting from the gap edges.  
A better algorithm in terms of requiring less total rotation (i.e. less energy consumption) has been proposed in \cite{wu15}. They have demonstrated that the total number of gaps and the total gap length is minimized using their algorithm.

%%%%%%%%%%%%%%%%%%%%%%%%%%
\section{Mathematical model}
\label{sec:mathmodel}

In this section, we first propose a novel discrete formulation for the line coverage problem. In this model, a number of co-linear points are considered as targets. This is in contrary to the conventional line coverage formulation in which a line is considered as a continuous two-dimensional line segment located in the region of interest \cite{prakash14, dash12}.  
We will discuss later on how the continuous line coverage problem  can be formulated in terms of the points coverage problem. 
%%%%%%%%%%%%
\subsection{Line coverage}% - discrete  version}
%In the discrete version, instead of a line segment,
In this problem, the objective is to cover $n_t$ co-linear targets, $T=\{ T_1, \cdots, T_{n_t} \}$, using  a minimum set of $n$ available sensors, i.e. $S=\{ s_1, \cdots, s_n\}$.
In this paper, we use the words `target' and `target points', interchangeably. 
Without loss of generality, targets are assumed to be located on the $x$ axis, where $ x(T_1) \le \cdots \le x(T_{n_t})$. Here, $x(T_i)$ represents the $x$ coordinate of the $i^{th}$  target.

% Sensors intervals 
 In this paper, we use line coverage in modeling the barrier-related problems. Hence, we project the sensor 2D coverage areas onto 1D intervals located on the $x$ axis. For the actual line coverage applications, the intervals would be the line segments of intersecting a line with the 2D sensor coverage areas instead of projected ones.  
 
 As depicted in Fig. \ref{fig:sensor_interval}, the coverage area of sensor $s_i$ is considered as a projected interval of $[u_i, v_i]$, regardless of the sensor being directional or omni-directional. We clip the sensor intervals if they extend before $T_1$ or after $T_{n_t}$. Furthermore, the sensors are sorted and numbered based on the start point of their intervals, i.e. $u_1 \le u_2 \cdots \le u_n$.  
In case of a tie, i.e. two or more sensors with equal $u_i$'s, we can assign a smaller number to the sensor with a smaller $v_i$.

\begin{figure}
\includegraphics[height=0.23\textheight,width=0.5\textwidth]{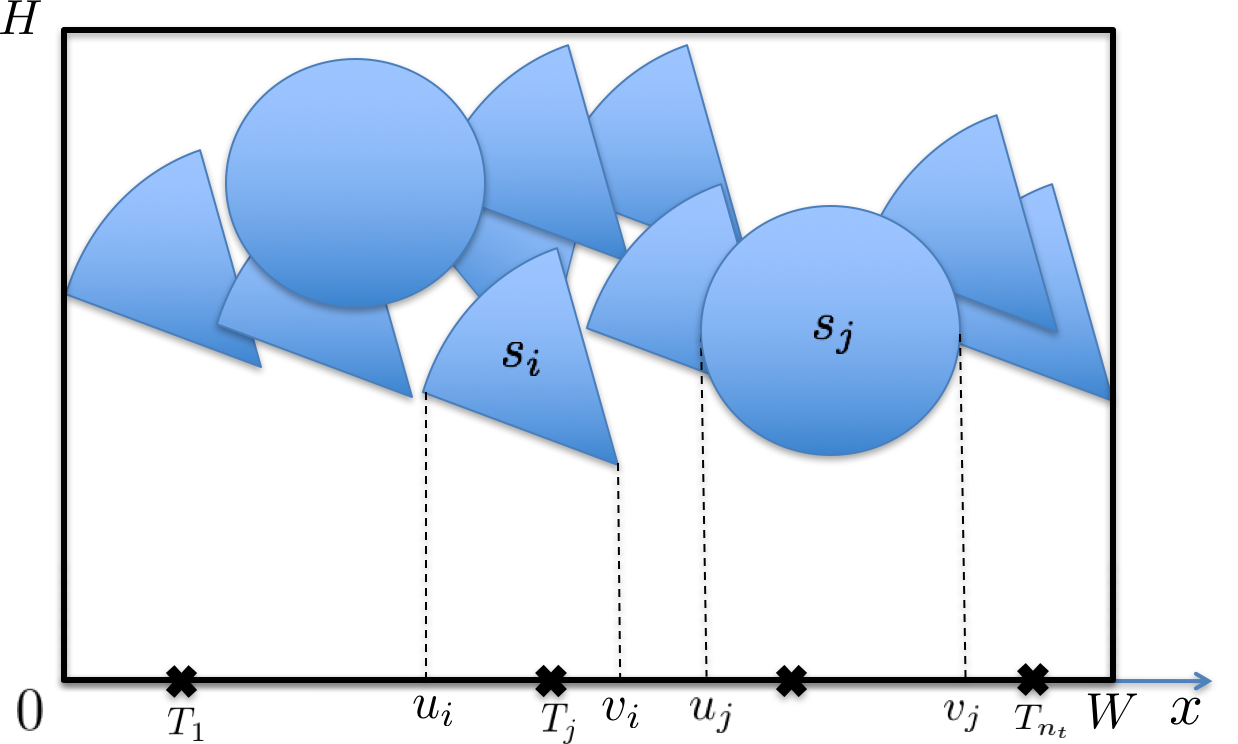}
\caption{Omni-directional/directional sensors coverage areas projection onto intervals on the $x$ axis}
\label{fig:sensor_interval}
\end{figure}

%____________________________________________________
\subsection{Optimal sensor selection for the line coverage}

In sensor selection for the line coverage problem, the objective is to select the \emph{minimum number of sensors} from the set $S$, so that we cover all the targets located on a single line.

%%%%%%%%%%%%%%%%%%
\subsection{The OGA algorithm}

The key point in our method is to exploit the existing \emph{order} of targets located on a single line. %and the start points of the sensor coverage intervals. 
The proposed algorithm selects the sensors in a greedy manner based on the targets order hence the name order-based greedy algorithm (OGA).
OGA starts from the far left target on the line, i.e. $T_1$. It continues increasing the coverage in each step by selecting the sensor with the maximum coverage, i.e. maximum number of covered targets on the right. 
In general discrete targets coverage, the selected sensors are not required to be overlapped. Distant non-overlapping sensors might be selected in case of far distance of consecutive targets. But, we do not skip any targets and cover them in order starting from the left hand-side.
OGA stops once it reaches the furthest target on the right, i.e. $T_{n_t}$.% or there is no more sensors in $S$ to select from. 

Here, we assume there exists a subset of sensors, which can cover all targets %the entire line (or all targets)
without any gaps, i.e. uncovered targets. 
In case of lack of enough sensors, i.e. having gaps, we add new sensors covering the gaps to $S$. The added sensors do not exist in reality, but help us use the proposed algorithms. In case of lack of enough sensors, there will be gaps after applying the algorithms, which should be mended by rotation of sensors or using mobile sensors as studied in the literature. In this paper, we consider a different scenario of gap mending, where gaps are due to sensor failure and other sensors are available to select from.
%{\hl{We define a gap as an interval, which cannot be covered by any sensors. In case of having gaps, we predetermine the gap intervals and add them as `virtual sensors' to $S$ to be able to use the proposed algorithms.
%}}

The details of the algorithm is listed in Algorithm \ref{alg:OGA}.
At each step, it selects the sensor ($s^g_c$) that covers the maximum number of targets to the right from $S^a_c$, i.e. set  of  sensors available to cover the current target ($T_c$). In case of tie, i.e. two sensors with equal number of covered targets to the right, the sensor with the greatest number of total covered targets (on the right and left) is selected.
The selected sensor is added to the set of selected sensors so far ($S^g$).
Then, the most right-hand side target ($T^c_{R}$) covered by $s^g_c$ is considered as the current target for the next step.  The set of sensors covering $T^c_{R}$ is denoted by $S_{T^c_{R}}$. 
The algorithm continues until the most right-hand side target ($T_{n_t}$) is covered. At this final step, $S^g$ and $N^g$ represent the final set of selected sensors and its cardinality, respectively.

%Let $S^g_{c}$ denote the set of sensors selected up to the current step (considering the $T_{c}$) in OGA. Moreover, assume $S^a_{i}$ represents a subset of $S$ that are available to cover target $T_i$. 
%Furthermore, $S_{c}$ denotes the set of sensors under consideration in the current step of the algorithm. The most right-hand side target in the current set and the set of sensors covering it are denoted by $T^c_{R}$ and $S_{T^c_{R}}$, respectively.
%Finally, $S^g$ and $N^g$ represent the selected sensors and number of them, respectively. % WRONG ==> We have $N^g=| S^g |$. 

%%%%%%%%%%%%%%%%%%%%%
% This way the length of  gaps are minimized???!!!

\begin{algorithm}
\SetAlgoLined
\KwData{$S, T$}
\KwResult{$S^g$}
%\Input{this}
%Start from target $T_1$, 
Set $T_{c} \leftarrow T_1$ and $S_{c} \leftarrow S^a_{1}$ \;
 \While{$T_{c} \ne T_{n_t}$}{
 % or $S_{c} \ne \emptyset$}{
Select $s^g_{c}$ from $S_c$ and determine $T^{c}_{R}$ \;
%In $S_c$, determine $s^g_{c}$ and $T^{c}_{R}$ \;% the sensor that covers the maximum number of targets to the right (covers $T^{c}_{R}$)\;
\If {$S^a_c$ does not cover $T^{c}_{R+1}$}{
Select $s^g_{c}$ from $S$ that covers  $T^{c}_{R+1}$
%From $S$, select the sensor covering $T^{c}_{R+1}$ with the maximum coverage to the right and call it $s^g_{c}$\;
}
   $S^g \leftarrow S^g \cup s^g_{c}$\;
   $S_{c} \leftarrow S_{T^c_{R}}$\;
}
 \caption{The OGA algorithm}
 \label{alg:OGA}
\end{algorithm}

%%%% OGA properties 
In the following, we first prove a Lemma regarding the properties of the selected sensors. For the Lemma, we define $S_i$ as the set of `selected' sensors for covering $T_i$. Using the Lemma, we prove a theorem on optimality of OGA in terms of selecting the minimum number of sensors.

%%%%%%%%%%%%%%%%%%%%
\begin{lemma}
\label{lemma:ordeed_union}
The set of selected sensors ($S_i$) possesses Markov property on the $x$ axis, i.e.
\begin{equation}
 | ( \bigcup_{k=1}^i S_k) \cap S_{i+1}| =| S_i  \cap S_{i+1}|
 \label{eq:markov}
\end{equation}
\end{lemma}

%%%%%%%%%%%%%%%%%%%%
\begin{proof}
As stated, the sensors and targets are sorted and we have $u_1 \le u_2 \cdots \le u_n$ and $x(T_1) \le \cdots \le x(T_{n_t})$. Furthermore, $S_i$'s might have overlap, i.e. share sensors in general.
%Let $S_i$ denote the set of `selected' sensors that are covering the target $T_i$. These sets could have intersection in general.
We prove by contradiction.
%It is easy to prove (\ref{eq:markov}) by contradiction.
 Assume (\ref{eq:markov}) does not hold and implies
\begin{equation}
\exists \quad s_l \quad \mbox{s.t.} \quad  s_l \in (\bigcup_{k=1}^{i-1} S_k \cap S_{i+1}) \quad \& \quad s_l \notin S_{i} 
\end{equation}
i,e. there is a sensor that covers at least one of the targets on the left-hand side of the $T_i$ as well as $T_{i+1}$, but does not cover $T_i$. This is impossible since sensor coverage intervals are contiguous and targets are ordered. This contradiction completes the proof.
\end{proof}

The Markov property helps us derive the local optimality property of the cost function in the sensor selection for $1$-coverage (Theorem \ref{thm:OGAisoptimal}) as well as $K$-coverage (Theorem \ref{thm:KOGA}).

%%%%%%%%%%%%%%%%%%%%%%%
\begin{theorem}
\label{thm:OGAisoptimal}
The OGA algorithm is optimal in selecting the minimum number of sensors.
\end{theorem}
%%%%%%%%%%%%%%
\begin{proof}
Define $C(T_i)$ as a cost function showing the number of selected sensors covering the targets from $1$ to $i$, i.e. $T_1, \cdots, T_i$. We have
\begin{eqnarray}
&C(T_{i+1})=  | \bigcup_{k=1}^{i+1} S_k| \nonumber\\
=  & | \bigcup_{k=1}^{i} S_k| + |S_{i+1}| - | ( \bigcup_{k=1}^i S_k) \cap S_{i+1}|
\end{eqnarray}
Using (\ref{eq:markov}), we can rewrite the cost function as
\begin{eqnarray}
 &C(T_{i+1})  = | \bigcup_{k=1}^{i} S_k| + |S_{i+1}| - | S_i  \cap S_{i+1}| \nonumber\\
 &= C(T_i)+ C_{i \rightarrow i+1}
\end{eqnarray}
where $C_{i \rightarrow i+1}= |S_{i+1}| - | S_i  \cap S_{i+1}|$ is the number of new sensors selected to cover $T_{i+1}$. As seen, the defined cost function, i.e. $C(\cdot)$, has the local optimality  property (i.e. Bellman's equation \cite{bellman} in dynamic programming literature). Hence, an optimal algorithm for selecting the minimum number of sensors will be optimal at every step. %Since $C(T_{i+1})$ and $C(T_i)$ are both optimal, $C_{i \rightarrow i+1}$ has to be minimized as well.
We provide the proof of optimality by induction. For $T_1$, both the optimal and the proposed greedy algorithm have a cost of $1$. If they are equivalent in terms of $C(T_i)$, then we prove their $C(T_{i+1})$'s are also equivalent. This is easy to show
since the greedy method minimizes $C_{i \rightarrow i+1}$ in every step and it completes the proof.
\end{proof}

%%%%%%%%%%%%%%%%%%%%%%%%%
The following theorem, analyzes the average order of complexity of OGA. 
\begin{theorem}
The average complexity of OGA is $\mathcal{O}(n)$.
\end{theorem}

\begin{proof}
%In the complexity analysis of both the continuous and discrete versions, 
We focus on the OGA algorithm steps and their order of computational complexity.
In the complexity analysis, we do not consider the complexity of finding the coverage relationship between the sensors and targets, i.e. which sensors covering which targets.

Starting from $T_1$, we have to find the sensor with the maximum $v_i$ among a maximum of $n_{com}$ sensors, where $n_{com}$ represents the average number of overlapping sensors covering a single point on the line segment. We perform such maximizations $\bar{N}^g$ times, where $\bar{N}^g \le n$ is the average total number of sensors selected by the OGA, hence $\bar{N}^g=\frac{n}{n_{com}}$. In conclusion, the algorithm average complexity is of order $\mathcal{O}\left(n_{com} \cdot \bar{N}^{g} \right)= \mathcal{O}(n)$.
\end{proof}

%%%%%%%%%%%%%%%%%%%%
%\begin{theorem}
%The average complexity of the OGA discrete version is $\mathcal{O}(n)$. % if $n_t < n$ or $\mathcal{O}(n \log (n \cdot n_t))$ otherwise.
%\end{theorem}
%%%%%%%%%%%%%%%
%\begin{proof}
%Similar to the previous proof, the algorithm performs an average number of $\frac{n_t}{n_{com}}$ instances of the maximum selection operation. In fact, if more than one target is covered by a single sensor, we need to perform less operations. All maximum selections are performed on sets with an average cardinality of $n_{com}$. Hence, the average complexity of the discrete version is $\mathcal{O}(n)$. 
%\end{proof}

%%%%%%%%%%%%%%%%%%%%%%%%%%%%%%%
\subsection{Continuous version of the line coverage problem}

A more conventional setting of the line coverage problem is depicted in Fig. \ref{fig:con_line}. As seen, a line segment from point $a$ to point $b$ is desired to be covered by sensors coverage intervals.
The sensor intervals are resulted from the projection of the sensor 2D coverage areas, i.e. sectors or circles, onto the target line.

% continuous to discrete conversion
Although our previous formulation possesses discrete targets, any continuous line coverage problem can be converted to an equivalent discrete one as depicted in Fig. \ref{fig:con2dis}. As seen, the equivalent targets are considered as the middle points of the sensor intervals.  
In this conversion, if the number of sensors is $n$, we will have a maximum of $2n-1$ targets in the case that no coincidences exists among the interval start or end points.

%%%%%%%%%%%%%%%%%%%%%%%%%%
\begin{figure}[!t]
\center
\includegraphics[height=0.13\textheight,width=0.42\textwidth]{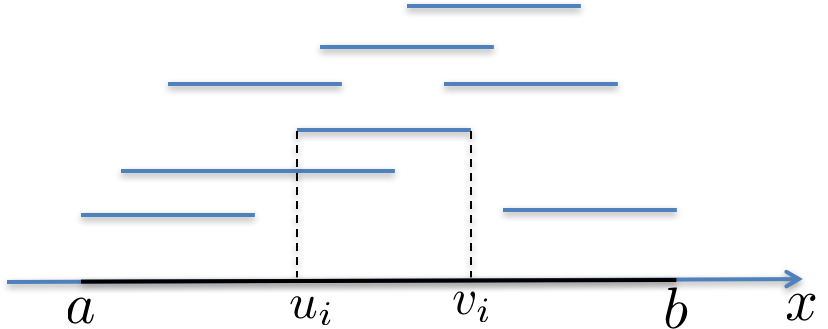}
\caption{The coverage areas of sensors projected onto intervals on the $x$ axis}
\label{fig:con_line}
\end{figure}
%%%%%%%%%%%%%%%%%%%%%%%%%%

% continuous OGA
%%%%%%%%%%%
Other than conversion to the proposed discrete model, we propose a dependent  algorithm, called continuous OGA, to solve the continuous problem. The details of the algorithm is depicted in Algorithm \ref{alg:con_OGA}.
It starts from Point $a$ and selects the sensor ($s^g_c$) with the largest coverage interval spread to the right ($v^g_c$), among the sensors available to cover it ($S^a_a$). 
In case of tie, i.e. two sensors with equal spread to the right, the sensor with the largest interval  is selected.
In the next step, the selection continues from $v^g_c$ and stops when $b$ is covered.

%%%%%%%%%%%%%%%%%%%%%%%%
\begin{algorithm}
\SetAlgoLined
\KwData{$S, [a, b]$}
\KwResult{$S^g$}
Start from point $a$, $S_{c} \leftarrow S^a_{a}$ \;
 \While{$v^g_{c} < b$}{
 % or $S_{c} \ne \emptyset$}{
 Select $s^g_c$ from $S_{c}$ and determine $v^g_c$ \;
% Among the sensors in $S_{c}$, find $s^g_c$ and determine the $v^g_c$ \;
 %the one that has the largest $v_i$, i.e. a sensor ($s^g_{c}$) with $v^g_{c}$ as its interval end \;
%  \eIf{condition}{
   $S^g \leftarrow S^g \cup s^g_{c}$\;
   $S_{c} \leftarrow S^a_{v^g_{c}}$\;
%  }
}
 \caption{The continuous version of the OGA algorithm}
 \label{alg:con_OGA}
\end{algorithm}

The average complexity of Algorithm \ref{alg:con_OGA} is of order $\mathcal{O}(n)$ and the proof is similar to that of Algorithm \ref{alg:OGA}.
%%%%%%%%%%%%%%%%%%%
%\section{The line coverage applications in wireless sensor networks}
Now, we demonstrate the applications of the proposed line coverage model in the weak $K$-barrier coverage and gap mending.  

%%%%%%%%%%%%%%
\section{Sensor selection for weak $K$-barrier coverage}
\label{sec:kbarrier}

%Assume we have a number of sensors, either omni-directional or directional, previously deployed in a region of interest, typically a rectangle of width $W$ and height $H$. 
%A sample scenario is depicted in Fig. \ref{fig:wbss}. It shows how the coverage area of omni-directional or directional sensors can be projected onto intervals on the $x$ axis in the weak barrier coverage problem.
The sensor selection problem for weak $K$-barrier coverage is to select the minimum number of sensors to form $K$ weak barriers. 
In the conventional weak $K$-barrier coverage, $K$ paths of weakly overlapped sensors should be formed. Weakly overlapped sensors refer to those overlapped after projection onto a line \cite{wang14}, i.e. sensors with overlapped coverage intervals. We refer the reader to the scenario depicted in Fig. \ref{fig:sensor_interval}.

We model this problem as a line coverage problem in which all the points on the coverage interval of interest ($[0,W]$) have to be covered by at least $K$ sensors. Since the sensor intervals are the results of orthogonal projection of 2D coverage areas onto the $x$ axis, the new problem is equivalent to the original weak coverage problem.
%As discussed, we project the sensors coverage areas onto the $x$ axis and create a number of intervals. After this conversion, the objective of the barrier coverage can be translated into covering the entire line segment from $0$ to $W$ on the $x$ axis.
%The resulted equivalent problem is a continuous version of the line coverage and can be solved either directly or by conversion to a discrete problem.
%%%%%%%%%%%%%%%%%%%%%%%%%%%%%
%\subsection{Sensor selection for weak K-barrier coverage}
We convert the continuous problem of $K$-coverage of $[0,W]$ to a discrete problem using the conversion  method discussed before. Once conversion is performed, the special case of $1$-coverage can be solved directly using OGA. For the general case of $K$-coverage, we propose a generalized version of OGA, called $K$-OGA.  %$K$-OGA is similar to OGA with the only difference that at each step, we select $K$ sensors with the largest $v_i$s. 

%%%%%%%%%%%%%%%%%%%
\subsection{The $K$-OGA algorithm}
$K$-OGA is a generalized version of OGA, which solves the weak $K$-barrier coverage problem optimally. 
The details of $K$-OGA is listed in Algorithm \ref{alg:KOGA}.
Assume $NC_s$ and $RS_s$ represent the set of NOT $s$-covered targets and the set of remaining  (not selected) sensors by the $s^{th}$ step of the algorithm, respectively. $K$-OGA starts with all targets in $NC_1$ and sets $RS_1=S$. It applies OGA to $1$-cover all the targets and removes the selected sensors from $RS_s$. Furthermore, it updates $NC_2$ since some targets might be $2$-covered by the overlapping sensors selected in the first step. The process continues and $K$-OGA attempts to increase the coverage of targets by one in each step using OGA. It stops when all targets are $K$-covered, i.e. $NC_K= \emptyset$.
 $S^g$ shows the set of selected sensors in Algorithm \ref{alg:KOGA}.
 In the next Theorem, we prove the optimality of $K$-OGA.

%%%%%%%%%%%%%%%%%%
\begin{algorithm}
\SetAlgoLined
\KwData{$S, T, K$}
\KwResult{$S^g$}
$s=1$\;
$RS_1 \leftarrow S$\;
Update $NC_{1}$\;
Apply OGA to $NC_1$\;
 \While{$NC_K  \ne \emptyset$}{
$s\leftarrow s+1$\;
Update $RS_{s}$ and $NC_{s}$\;
Apply OGA to $NC_{s}$\;
}
 \caption{The $K$-OGA algorithm}
 \label{alg:KOGA}
\end{algorithm}
%\begin{enumerate}
%\item Select sensors using OGA to provide $1$-coverage
%\item $count \leftarrow 1$
%\item Remove the selected sensors from the list of available sensors
%\item $count \leftarrow count+1$ 
%\item if $count=K$, stop
%\item Form a new list of targets, which are not $'count'$ covered yet
%\item For the new list of targets, perform a new round of OGA to provide $count+1$ coverage
%\item Go to step $3$
%\end{enumerate}

%%%%%%%%%%%%%%%%%%%%%%%%%%%%%
%%\begin{proposition}
%%\label{prop:bottleneck}
%%In the optimal $K$-coverage, we must have at least {\bf{one}} `bottleneck target' for which $|S^*_K| = K$.% otherwise we can reduce the number of sensors and still maintain the $K$-coverage. 
%%\end{proposition}
%%
%%\begin{proof}
%%We prove the theorem by contradiction. Hence, assume $|S^*_i| \ge K+1 \quad \forall i$. Among the optimally selected sensors, there is no sensors not covering any targets. Otherwise, the sensor selection will not be optimal in terms of number of sensors. 
%%Assume target $j$ has the least order of coverage among all sensors, say $(K+\Delta K)$-coverage. 
%%Randomly remove $\Delta K$ sensors from $S^*_{j}$. The $j^{th}$ target will remain the least covered one, i.e. $K$-covered since other targets will be less affected. In conclusion, we do not lose $K$-coverage. Hence, any optimal selection will have at least one bottleneck target.
%%\end{proof}

%%%%%%%%%%%%%%%%%%%%%%%%%%%%%%
\begin{theorem}
\label{thm:KOGA}
%The greedy method proposed for weak K-barrier coverage in \cite{zhibo} (Algorithm 2 in their paper) is optimal.
$K$-OGA is optimal in selecting the minimum number of sensors for $K$-coverage.
\end{theorem}

\begin{proof}
The proof has two parts. First, we find an optimal sensor selection method based on the local optimality property of the cost function. Then, we prove that $K$-OGA is equivalent to the optimal method and selects the same sensors.

In the first step, we find an optimal method. Define $C^K(T_i)$ as the number of sensors selected for the $K$-coverage of targets (from the left in a sorted list of targets) up to the target $i$. Similar to the proof of Theorem \ref{thm:OGAisoptimal}, we can see that $C^K(T_i)$ also has local optimality property.
Hence, we can exploit the property (order) in $K$-covering the targets. In fact, the optimal method can start from Point $a$ and selects $K$ sensors (forming $S_a$) with the largest $v_k$'s to minimize $C^K_{a \rightarrow 1}$. Next, it considers the next target ($2^{nd}$ one) and attempts to select the same sensors as in $S_a$ to keep the cost low. For the new sensors, it selects the longest ones to the right.
It continues this procedure until all targets are $K$-covered or there is no sensors to select from. 

In the next step, we prove by induction on $K$ that $K$-OGA selects the same sensors as the optimal method. In case of $K=1$, $K$-OGA is the same as OGA, which is a special case of the mentioned optimal method, i.e. for $K=1$. Now, assume that both methods select the same sensors for $K=m$. We have to prove that they select the same sensors for $K=m+1$. This is easy to prove by noticing that both methods are order-based and start from the point $a$ (or the first not $m+1$-covered point). For $K=m+1$, the optimal method selects the $m+1$ longest sensors that are covering $a$. The first $m$ longest sensors are the same as those selected for $K=m$, which are also selected by $K$-OGA. The $(m+1)^{th}$ sensor selected by the optimal method is also selected by $K$-OGA to provide $m+1$ coverage.  For the rest of targets, they both skip  the already $m+1$-covered targets to reach the first non-covered point. This procedure continues and both methods select the same sensors. 
\end{proof}

%%%%%%%%%%%%%%%%%%%%%%%%%
\begin{theorem}
The average complexity of the $K$-OGA algorithm is $\mathcal{O}(Kn)$.
\end{theorem}

\begin{proof}
The first step of the algorithm runs OGA and has an average complexity order of $\mathcal{O}(n)$. After the first step, the average number of targets that need $2$-coverage reduces to $N_t \times (1-\frac{1}{n_{com}})$, where $N_t=2n-1$ if no sensor intervals are the same.
In the $i^{th}$ step, the average number of targets required to be considered for $i$-coverage is $N_t \times(1-(\frac{1}{n_{com}})^i)$, which gets closer to $N_t$ as $i$ increases. Hence, the average number of targets can be approximated by $N_t=2n-1$ in every step. Hence, the average complexity of $K$-OGA is approximately equal to $K$ times that of OGA, i.e.  $\mathcal{O}(Kn)$.
\end{proof}

% compare with \cite{wang14}
In \cite{wang14}, two algorithms, called optimal and greedy versions of Min-Num-Mobile(k) algorithm are proposed for selecting the minimum number of static/mobile sensors for $K$-barrier coverage.
%Two algorithms, i.e. the optimal and greedy versions of Min-Num-Mobile(k), are proposed in \cite{wang14} for the minimum number of static/mobile sensor selection in the $K$-barrier coverage problem. 
The complexity of both algorithms is of order $\mathcal{O}(kn^2)$.
The proposed algorithms can be applied in our problem if the assumed mobile sensors are replaced with the gap sensors introduced here. We use their greedy algorithm in Section \ref{sec:simres} as a benchmark. It should be noted that the complexity of $K$-OGA is considerably less than their method due to the exploitation of order in the 1D modeling.

%%%%%%%%%%%%%%%%%%%%%%%%%%%%%%
\begin{figure}[t!]
\center
\includegraphics[height=0.12\textheight,width=0.42\textwidth]{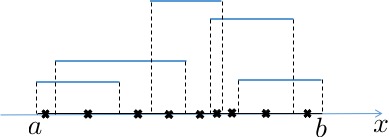}
\caption{Continuous to discrete conversion in the line coverage problem}
\label{fig:con2dis}
\end{figure}
%%%%%%%%%%%%%%%%%%%%%%%%%%%%%%
%%%%%%%%%%%%%%%%%
%\begin{theorem}
%The greedy method proposed for K-barrier coverage in \cite{wang14} (Algorithm 2 in the paper) is optimal for the 'weak' $K$-barrier problem.
%\end{theorem}
%
%\begin{proof}
%As stated, the proposed method builds a weighted graph first. Afterwards, it selects $K$ vertex-disjoint shortest paths one after another using Dijkstra Algorithm with node removal. Finding the shortest path in each iteration, is equivalent of the OGA algorithm since it minimizes the length of path, i.e. number of sensors. Hence, the proposed method is equivalent of $K$-OGA, which is proved to be optimal. 
%\end{proof}

%%%%%%%%%%%%%%%%%%%%%%%%
%\begin{lemma}
%The complexity of the selection method proposed in \cite{zhibo14} when applied to this problem is $\mathcal{O}\left( n(\log{n}+ n_c) \right)$.
%\end{lemma}
%
%\begin{proof}
%$|E|=n \cdot n_c$
%\end{proof}

%%%%%%%%%%%%%%%%%%%%%%%%
\section{Efficient sensor selection for weak barrier gap mending in case of sensor failure}
\label{sec:failure}

%\hl{To the best of our knowledge, this problem has not been addressed in the literature.} 
In practice, for weak $K$-barrier coverage, sensors are first deployed  (randomly or strategically) and then possibly modified (selected/rotated) to form $K$ weak barriers.
The deployed/modified sensors might fail to provide coverage after a while  due to mechanical failure,  drain of battery, etc. One of the strategies to mend the resulted barrier gaps is to select the minimum number of new sensors from the remaining ones, i.e. not selected initially.

One trivial algorithm for sensor selection is to re-apply a new instance of OGA algorithm to the set of all sensors except the failed one to achieve an optimal solution. Here, we propose another algorithm, called LOGM (locally optimal gap mending), which can locally mend the gap without running the OGA over the entire set of remaining sensors. Here, we assume the failed sensors were selected by OGA after deployment. Hence, the failed sensors are not any randomly selected ones.

\subsection{The LOGM algorithm}
Let $[u^*, v^*]$ denote the coverage interval of a missing sensor. After sensor failure, still some parts of this interval might be covered by the previously selected sensors. Hence, the actual gap interval is $[u_G, v_G]$ as depicted in Fig. \ref{fig:sensorfailure} (If not, $u_G=u^*$ and/or $v_G=v^*$). In the depicted scenario, $v_G=u_{i^*+1}$. 
It should be noted that both the start ($u_G$) and end ($v_G$) of a gap could be end or start of a sensor interval coverage depending on where the gap exists.
The failed sensor interval is depicted as a red dashed line segment.

The LOGM performs sensor selection similar to OGA. In fact, it adopts the sensors previously selected by the previous instance of OGA.
The only difference between LOGM and the new instance of OGA  is that LOGM starts new sensor selection from $u_G$ and stops once $v_G$ is covered. Hence, LOGM is local and runs faster compared to a new instance of OGA. But, it might result in a non-optimal number of sensors since the newly selected sensors might have unnecessary overlaps with the adopted sensors.

Consider Fig. \ref{fig:sensorfailure2}, which includes more details regarding the selected sensors by both algorithms. The blue solid lines represent the sensors initially selected by the previous instance of OGA. These sensors are completely adopted by LOGM. The black dotted ones are the sensors that were not selected initially, but selected by both LOGM and OGA  to mend the gap. LOGM stops after selecting the $s_k$ since the gap is mended by then. However, OGA continues the sensor selection until it covers Point $b$.

%Let $S^x(\alpha, \beta)$ denote the set of sensors selected by algorithm `$x$' to cover the interval $[\alpha, \beta]$ in  a (continuous) $K$-barrier coverage problem.
Assume $S^l(a, \alpha)$ and $S^g(a, \alpha)$ represent the selected sensors by  LOGM and OGA, respectively, to $1$-cover $[a, \alpha]$.
In the next lemma, we prove that $|S^g(a, \beta)|$, i.e. the number of sensors selected by OGA to cover $[a, \beta]$, is an increasing function of $\beta$.

\begin{lemma}
\label{lem:increase}
 If $\beta_1 < \beta_2$,  then $|S^g(a, \beta_1)|  \le |S^g(a, \beta_2)|$.
\end{lemma}

\begin{proof}
OGA starts sensor selection from point $a$. Denote the sensor  with the largest $v$ in $S^g(a, \beta_1)$ by $s_{max}$. If $s_{max}$ covers $\beta_2$, no more sensors are required to be selected, i.e. $|S^g(a, \beta_1)|  = |S^g(a, \beta_2)|$. Otherwise, more sensors must be selected until $\beta_2$ is covered, i.e. $|S^g(a, \beta_1)|  < |S^g(a, \beta_2)|$.
\end{proof}

Based on Lemma \ref{lem:increase}, it is straight forward to conclude that $|S^g(\alpha_1, \beta)|  \ge  |S^g(\alpha_2, \beta)|$, if $\alpha_1 < \alpha_2$.
In the next theorem, we prove that LOGM selects at most one sensor more than the optimal, i.e. a new instance of OGA applied from the beginning to mend the gap.

%%%%%%%%%%%%%%%%%%%%%%%%
\begin{figure}
\center
\includegraphics[height=0.12\textheight,width=0.45\textwidth]{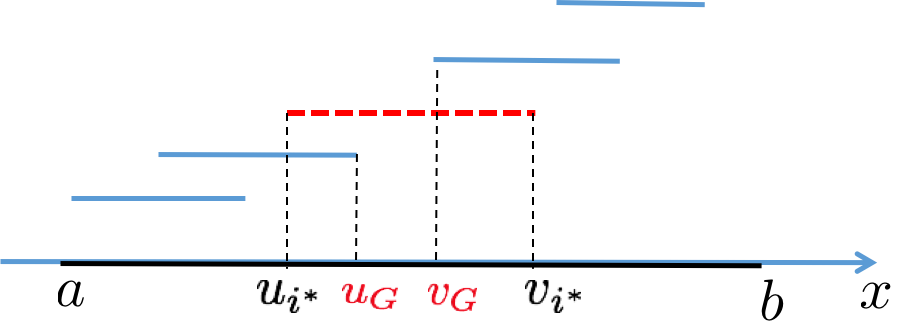}
\caption{A sensor failure scenario}
\label{fig:sensorfailure}
\end{figure}

%%%%%%%%%%%%%%%
\begin{theorem}
\label{thm:LOGM_max_1}
In the single-gap case, LOGM selects at most 1 sensor more than the optimal approach, i.e. $n_{OGA} \le n_{LOGM} \le n_{OGA}+1$.
\end{theorem}

%%%%%%%%%%%%%
\begin{proof}
The OGA method for gap mending is optimal hence $n_{OGA} \le n_{LOGM}$ is obvious. 
We have
\begin{equation}
S^l(a, v_G)=S^g(a, v_G) 
\end{equation}
since LOGM adopts the previous and new sensors selected by OGA to cover $[a, v_G]$.
As seen in Fig. \ref{fig:sensorfailure2}, $s_k$ is the last sensor selected by LOGM to mend the gap. Other sensors on the right, previously selected by OGA, including $s_{i^*+1}$ and $s_{i^*+2}$ are adopted by LOGM. 

We always have $v_k < v_{i^*+1}$, otherwise, $s_k$ would have been selected by OGA instead of $s_{i^*+1}$.
Now, consider the scenario depicted in Fig. \ref{fig:sensorfailure2}, where $v_k \ge u_{i^*+2}$. %$v_k > v_{i^*}$. 
It is the worst case for LOGM in terms of the location of $v_k$. The reason is $s_k$ along with  $s_{i^*+2}$ (or other sensors selected by OGA in a general scenario) completely cover $s_{i^*+1}$ and make it redundant for coverage. In other words, $s_{i^*+1}$ is adopted by LOGM and increases the cost, while it does not provide new coverage.
% the entire interval of $[v_G, v_k]$ from the coverage interval of $s_{i^*+1}$, i.e. $[v_G, v_{i^*+1}]$, overlaps with $s_k$ and becomes redundant.

Using Lemma \ref{lem:increase}, it is obvious that
\begin{equation}
\label{eq:OGA_is_less}
|S^g(v_k, b)| \le |S^l(v_{i^*}, b)|.  
\end{equation}
and
%This is obvious since the set of points on the $x$ axis that are covered by $S^g(v_k, b)$
%are a subset of that of $S^l(v_{i^*}, b)$. 
%Using the same argument results in
\begin{equation}
\label{eq:OGA_is_more}
|S^l(v_{i^*+1}, b)|  \le |S^g(v_k, b)|.  
\end{equation}
Based on (\ref{eq:OGA_is_less}), OGLM might result in selecting more sensors compared to the optimal. On the other hand,  (\ref{eq:OGA_is_more}) states that it is impossible for LOGM to select more sensors after (i.e. on the right-hand side of) $s_{i^*+1}$, compared to OGA. Hence, the worst case for LOGM happens when $s_{i^*+1}$ is completely covered by $s_k$ and $s_{i^*+2}$ in the scenario depicted in Fig. \ref{fig:sensorfailure2}.
 In this case, OGA does not select $s_{i^*+1}$, but the number of other sensors is the same as LOGM. Hence, in this case, $|S^l(v_G, b)| = | S^g(v_G, b)|+1$.

Finally, since $|S^l(a, v_G)|=|S^g(a, v_G)|$ and $|S^g(v_G, b)| \le |S^l(v_G, b)| \le |S^g(v_G, b)|+1$ in general, LOGM selects at most $1$ sensor more than OGA (i.e. the optimal method). 

%Assume that the most right-hand side sensor selected by LOGM %or the new OGA
%is $s_k$ (with $v_k$ as its end). LOGM stops after the selection of this sensor since the gap is mended. 
%%Hence, it adopts the next sensor selected by the previous OGA, i.e. $s_{i^*+1}$. 			
%Next, it adopts the next sensors selected by the old OGA, which provide optimal coverage starting from $v_{i^*}$.
%The new OGA selects the next sensors starting from $v_k$. The selected sensors are the same in both methods up to this point since they follow the same strategy.
%We have $v_k > v_{i^*}$ or $v_k = v_{i^*}$.
%If $v_k > v_{i^*}$, there is a possibility that LOGM selects more sensors compared to the new OGA (i.e. optimal gap mending method).
%On the other hand, we have $v_k < v_{i^*+1}$, otherwise its sensor would have been selected by the old OGA.
%The first sensor that LOGM adopts after $s_k$ is $s_{i^*+1}$. After this sensor, it selects the minimum number of sensors, which provide coverage after $v_{i^*+1}$. 
%The number of these sensors is necessarily less than or equal to that of the new OGA method selected for coverage after $v_k$ since $v_k < v_{i^*+1}$.
%%In summary, both algorithms select the same number of sensors till the selection of $v_k$.
%The worst case for the LOGM occurs when $v_k= v_{i^*+1}$. In this case, both methods select the sensors starting from $v_k$, which results in the same number of sensors except for $s_{i^*+1}$. Hence, in the worst case, the LOGM selects one more sensor compared to the new OGA.

\end{proof}

%%%%%%%%%%%%%%%%%%%%
\begin{figure}
\center
\includegraphics[height=0.17\textheight,width=0.5\textwidth]{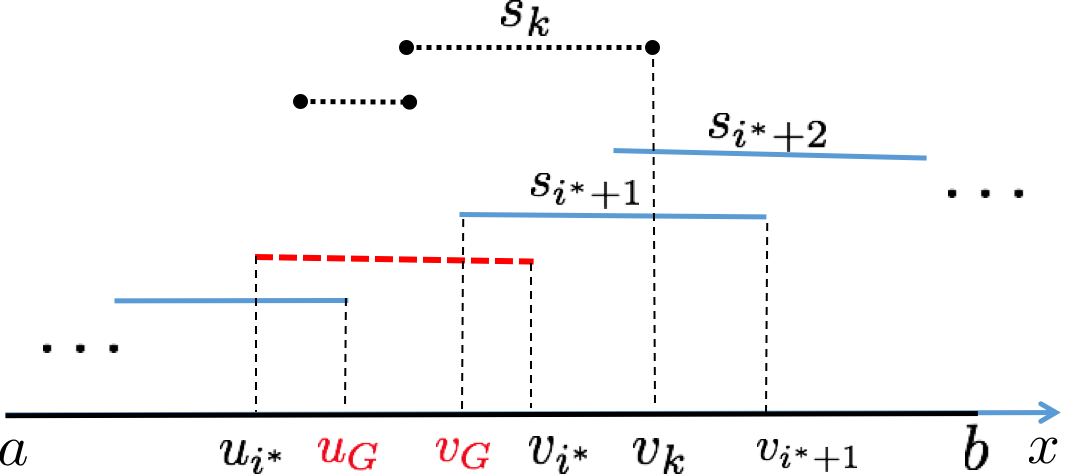}
\caption{A more detailed version of the sensor failure scenario depicted in Fig. \ref{fig:sensorfailure}}
\label{fig:sensorfailure2}
\end{figure}
%%%%%%%%%%%%%%%

\begin{theorem}
\label{thm:multi_gap}
In case of having $m$ gaps, LOGM selects at most $m$ sensors more than the optimal approach, i.e. $n_{OGA} \le n_{LOGM} \le n_{OGA}+ (2m-1)$.
\end{theorem}
%%%%%%%%%%%%%

\begin{proof}
Consider the case of $m=2$, where we have two gaps, $G_1$ and $G_2$.
A sample scenario is depicted in Fig. \ref{fig:multi_gap_scenario}. If there is overlap between the gaps, i.e. failure of neighboring sensors, the gaps are connected and the problem can be considered as a single gap one, where Theorem \ref{thm:LOGM_max_1} holds. If there is no overlap, the gaps are necessarily apart by at least one OGA sensor. OGA sensors are the ones initially selected by OGA.

\ignore{
Consider a new arbitrary point between the two gaps and call it $p_1$. Based on Theorem \ref{thm:LOGM_max_1}, the number of sensors to mend $G_1$ and cover all areas in the interval $[a, a_1]$ is one more for LOGM as compared to OGA. The same applies to $[a_1, b]$, which now has one gap ($G_2$). Hence the maximum possible number of extra sensors will be 2. 
It is straight forward to generalize the discussion to the case of $m$ gaps and conclude that LOGM selects at most $m$ sensors more than OGA, i.e. one extra sensor per gap. It completes the proof.
}
%___________________________________

%_____________________________________
%\ignore{
Using a discussion similar to that of the proof of Theorem \ref{thm:LOGM_max_1}, we conclude that the selected sensors for covering $[a, v_k]$ are the same in both OGA and LOGM. Furthermore, LOGM might select one extra sensor to cover $[v_k, v_{i_1^*+1}]$.
%In other words, the sensors selected for covering both left and right ends of the barrier, i.e. intervals distant from the gaps, are the same in both methods.
%Hence, we only need to investigate the sensors selected for mending the gaps ($G_1$ and $G_2$) and the intervals between them.

For the interval between the gaps, i.e.  $[v_{i_1^*+1}, u_{G_2}]$, 
LOGM adopts the previous OGA sensors. Hence, $|S^l(v_{i_1^*+1}, u_{G_2})|$ is optimal. As seen in the figure, the most right hand-side sensor selected by OGA to cover $[v_{i_1^*+1}, u_{G_2}]$ is $s_j$.
We have $v_j \ge u_{G_2}$ since $s_j$ must cover $u_{G_2}$. 

The key point in comparing the cost of methods between the gaps is that LOGM continues optimal selection from $u_{G_2}$, while OGA performs that from $v_j$. Therefore,
\begin{equation}
|S^l(u_{G_2}, v_{G_2})| \ge |S^g(v_j, v_{G_2},)|.
\end{equation}
On the other hand, the difference cannot be more than 1. The reason is LOGM performs optimal selection for gap mending, i.e. covering $[u_{G_2}, v_{G_2}]$. This optimal selection cannot be worse than that of OGA, i.e. adding $s_j$ to $S^l(u_{G_2}, v_{G_2})$ as an extra sensor to cover $[u_{G_2}, v_j]$ in addition to adopting OGA sensors selected to cover $[v_j, v_{G_2}]$.

%LOGM is in a better position based on a  similar discussion in the case of single gap. In fact,
% is in a better position compared to OGA in covering $[v_{i^*+1}, b]$. In the current case of having two gaps, LOGM 
%\begin{equation}
%|S^l(v_{i_1^*+1}, u_{G_2})| \le |S^g(v_k,u_{G_2})|.
%\end{equation}
%since $v_k < v_{i^*+1}$ as discussed. 
%Hence, it is impossible for OGA to outperform LOGM in between the gaps.
% For mending the gaps, consider the intervals $[u_{G_1}, v_{i_1^*+1}]$ and $[u_{G_2}, v_{i_2^*+1}]$. Based on a discussion similar to that of Theorem \ref{thm:LOGM_max_1}, one extra sensor per gap might be selected by LOGM.

After mending $G_2$, there is a possibility that LOGM selects one more sensor compared to the optimal due to redundancy of $s_{i_2^*+1}$ based on a discussion similar to the case of $s_{i_1^*+1}$.

In the case of $m$ gaps, we can decompose the space into the non-overlaping intervals:
$[a, u_{G_1}], (u_{G_1}, v_{G_1}], (v_{G_1}, u_{G_2}], (u_{G_2}, v_{G_2}], \cdots, (v_{G_m}, b]$. There is a total of 2m+1 intervals. For the first two intervals, both methods select the same set of sensors, hence no gain in OGA. In the rest of the intervals, LOGM may fall behind with one sensor per interval hence the total maximum will be 2m-1. 
%How much do we fall behind in each interval?

%In conclusion, a maximum of one extra sensor  for covering the interval between the gaps and another extra sensor per gap might be selected by LOGM compared to the optimal. It is straight forward to generalize the argument to the case of $m$ gaps. In the case of $m$ gaps, LOGM selects at most $m$ extra sensors for mending the gaps and $m-1$ extra sensors for covering the intervals between the neighboring gaps. Hence, it selects at most $m+ (m-1)= 2m-1$ sensors more than OGA.
%}

\end{proof}

%%%%%%%%%%%%%%%%%%%%%%%%
\begin{figure}
\center
\includegraphics[height=0.15\textheight,width=0.5\textwidth]{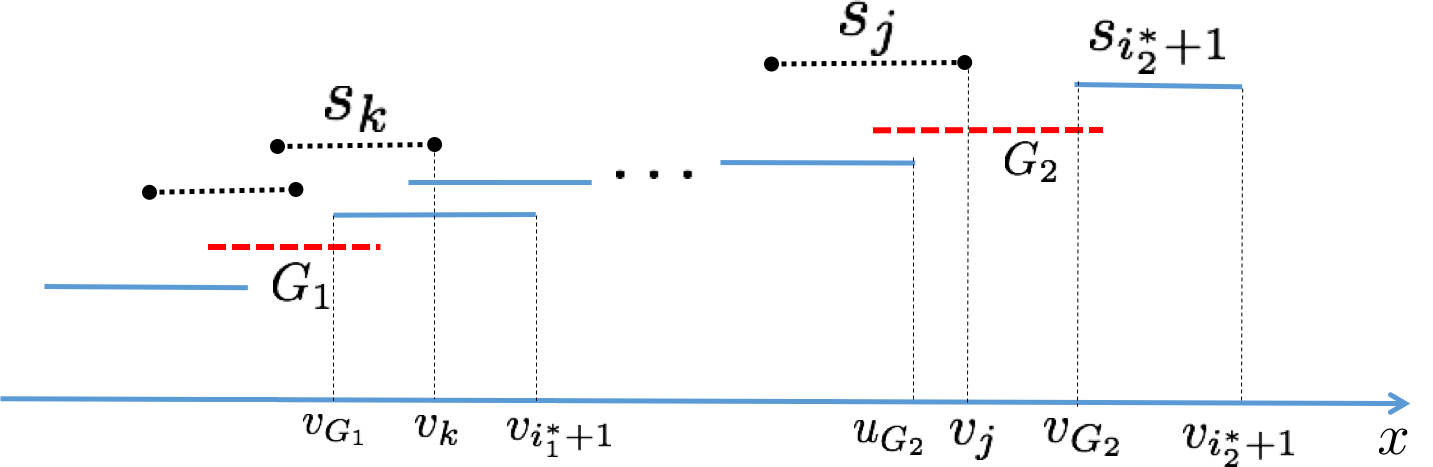}
\caption{A two-gap scenario}
\label{fig:multi_gap_scenario}
\end{figure}

%%%%%%%%%%%%%%%%%%%%%%%%%%%%
\section{Simulation results}
\label{sec:simres}

In this section, we confirm the optimal/efficient performance of the proposed algorithms via extensive simulations.

%%%%%%%%%%%%%%%%%%%%%%%%%%%%%%%%
\subsection{Sensor selection for weak 1-barrier coverage}
In this scenario, the performance of OGA is compared with a greedy benchmark in terms of the provided coverage vs. the number of selected sensors. %Both algorithms
The benchmark algorithm is inspired by a well-known approach  in designing greedy algorithms in the sensor networks literature \cite{fusco09, wang14}.
The algorithm updates a set of uncovered targets, i.e. the set of targets uncovered until the current step. It starts by selecting the sensor with the maximum coverage in $S$, i.e. the one that has the maximum number of covered targets. Then, the covered targets are removed from the set of uncovered targets and a new selection is performed based on the maximum coverage. The procedure continues until all of the targets are covered (barrier is formed) or there are no sensors to select from.

%%%%%%%%%%%%%%%%%%%%
\begin{figure}%[!t]
\centering
\subfloat[ Scenario 1]{\label{fig:scen1}\includegraphics[width=0.5\textwidth]{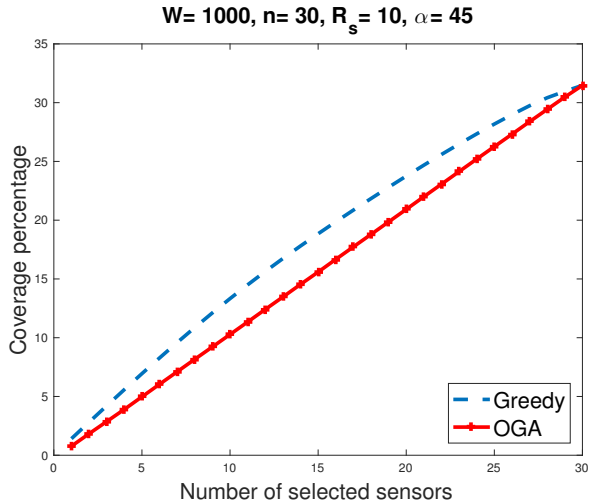}}\\
\label{fig:thirty}
\subfloat[ Scenario 2]
{\label{fig:scen2}\includegraphics[width=0.5\textwidth]{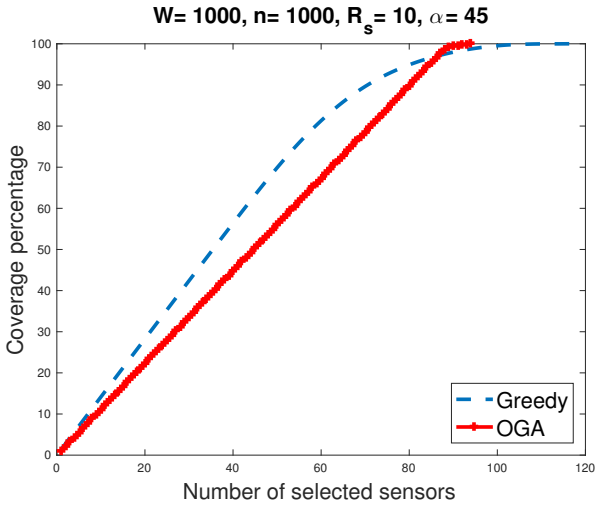}}\\
\subfloat[ Scenario 3 ]{\label{fig:scen3}\includegraphics[width=0.5\textwidth]{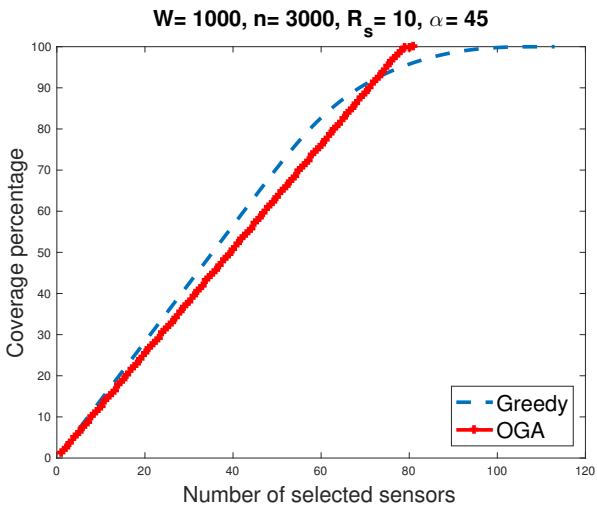}}\\
\caption{Performance comparison between OGA and the greedy benchmark in terms of the provided coverage versus the number of selected sensors}
\label{fig:coverage}
\end{figure}

%%%%%%%%%%%%%%%%%%%%%% 

% Coverage area and deployment
We consider a rectangle with a width and height of $1000$ and $10$ meters, respectively, as the coverage region of interest. In the following simulations, different numbers of (available) sensors are considered in 3 distinct scenarios.
The radius ($R_s$) and field-of-view ($\alpha$) of all sensors are the same.
However, the direction of each sensor is selected uniformly at random from $[0, 360)$.
Furthermore, the sensor deployment is line-based, where sensors are normally deviated from a straight line with a mean of zero and standard deviation of 10 meters along the $x$ axis. The deviation in the $y$ direction does not affect the weak barrier coverage performance in our 1D model. 

% results comparison
Fig. \ref{fig:coverage} depicts the provided (cumulative) coverage percentage versus the number of selected sensors in OGA and the greedy benchmark.  Coverage percentage is the portion of the $[0,W]$ interval covered by the selected sensors.  For instance, in Fig. 7.a, for the (first) $10$ selected sensors, the coverage percentage is $12$ and $10$ in the greedy method and OGA, respectively. The selected sensors are different in general as the two algorithms take different strategies for sensor selection.
 As seen, there is an intersection point, where the two curves intersect besides the origin. Before this point, the greedy benchmark outperforms OGA since it provides greater coverage by selecting every sensor. After the intersection, OGA outperforms the benchmark and maintains the lead until $100 \%$ coverage is achieved (only in gap free scenarios). Hence, the intersection point is a turning point. Although OGA starts by selecting less efficient sensors in terms of coverage, it ends up with smaller total number of selected sensors for complete coverage and outperforms the benchmark in total. 

% compare scenarios
In Fig. \ref{fig:coverage}, from Scenario 1 (Fig. \ref{fig:scen1}) to Scenario 3 (Fig. \ref{fig:scen3}), the number of available sensors is increased from $30$ to $3000$ and the achievable coverage percentage changes from about $30 \%$ to $100 \%$. As seen, OGA selects a total of 32, 95 and 82 sensors in Scenario 1 to 3, respectively, while the greedy method selects 32, 118 and 115 sensors in the corresponding scenarios. Hence, OGA always selects equal or smaller number of sensors compared to the greedy benchmark for the same amount of  achievable coverage.
Furthermore, the greater the number of available sensors, the better OGA performs compared to the greedy benchmark, i.e. the difference in the number of selected sensors (cost difference) increases.
%In Scenario 1 (Fig. \ref{fig:scen1}), the number of available sensors is greater compared to Scenario 2 (Fig. \ref{fig:scen2}). As seen, this resulted in a smaller intersection point and a greater difference in the cost (\# of selected sensors) of algorithms.
%In Scenario 3 (Fig. \ref{fig:scen3}), where the sensing radius of the sensors is smaller compared to Scenario 2, the intersection point is located at a smaller coverage percentage, i.e. OGA reaches the benchmark coverage by selecting less sensors. Furthermore, the cost difference is also greater in Scenario 3.  

% intersection and difference
The intersection point and the cost difference for different number  of available sensors are depicted in Fig. \ref{fig:diff_intersec}.
It is assumed that the width of the coverage region is $1000$ meters. Each sensor  has a radius of $10$ meters and a field-of-view of $90 \degree$. 
%The maximum coverage interval of each sensor is achieved if it is oriented exactly upwards, i.e. $\beta=90$, where $\beta$ is the sensor orientation angle. In this case, the  
As seen, when the number of available sensors is too small (high probability of having gaps), both algorithms select the majority of sensors (almost the same sensors) to maximize the coverage. Hence, the intersection is almost equal to the number of available sensors and happens at the top right corner as depicted in Fig. \ref{fig:scen1}.
As the number of available sensors increases, the intersection also increases until a turning point, where there are enough sensors for OGA to outperform the greedy algorithm. At this point, the intersections starts to decrease and the cost difference increases.
%When the number of sensors is large enough so that there are many large intervals for the benchmark to select from, it becomes less efficient, the intersection moves to the left (becomes smaller) and the cost difference increases.

%%%%%%%%%%%%%%%%%%%%%%%%%%
\begin{figure}[!t]
\centering
\includegraphics[width=0.5\textwidth]{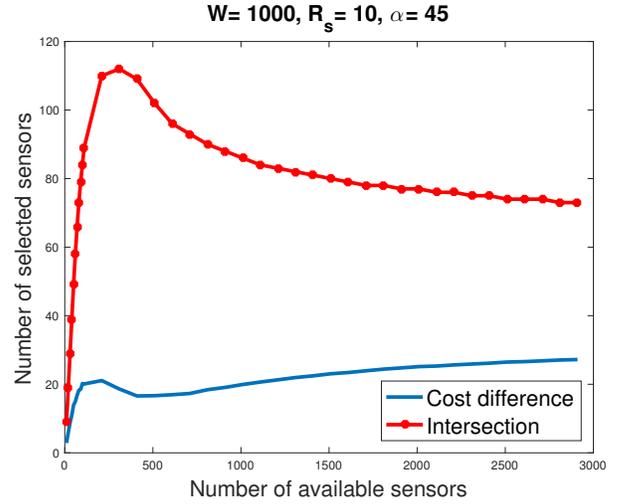}
\caption{The intersection point and the cost different between OGA and the greedy benchmark}
\label{fig:diff_intersec}
\end{figure}

%%%%%%%%%%%%%%%%%%%%%%%%%%%%%
%\begin{figure}[!t]
%\centering
%\includegraphics[width=0.5\textwidth]{figure/diff.eps}
%\caption{123}
%\label{fig:diff}
%\end{figure}
%
%%%%%%%%%%%%%%%%
%\begin{figure}[!t]
%\centering
%\includegraphics[width=0.5\textwidth]{figure/intersec.eps}
%\caption{123}
%\label{fig:intersec}
%\end{figure}

%%%%%%%%%%%%%%%%%%%%%%%%%%%%
\subsection{Sensor selection for weak $K$-barrier coverage}

% benchmark K-barrier
As stated, a graph-based method is proposed in \cite{wang14} for solving the strong/weak K-barrier coverage problems. In this method, the graph nodes represent the sensors. The nodes are connected via an edge if they have overlap. In the case of weak barrier coverage, the overlap is defined as overlap between the intervals projected onto the $x$-axis as we defined here. Hence, sensor coverage areas could be non-overlapping in the real world 2D deployment, while their projections on $x$-axis are. 
It is proved that the existence of $K$ strong/weak barriers is equivalent to that of $K$ vertex-disjoint paths in the proposed graph. In order to find the possible $K$ paths, a greedy algorithm is proposed, which selects $K$ paths one after another using the modified Dijkstra's algorithm \cite{dij}. The algorithm removes the selected node after the selection of each path. 
We use this greedy algorithm as a benchmark for the weak $K$-barrier coverage problem.

 In our problem, all sensors are stationary. In the case of having gaps, i.e. lack of at least $K$ sensors to cover each target, we use the new gap sensors introduced before. In \cite{wang14}, fictitious mobile sensors are assumed to cover these gaps.
 In order to use their method as the benchmark, we use new non-existing gap sensors instead of mobile ones. Hence, a (multi-) mobile sensor(s) covering a single gap in their method is replaced with a single gap sensor in our implementation. In fact,  the graph edge weights are changed but, the same algorithm is used to find $K$-barriers, i.e. $K$ vertex-disjoint shortest paths.

Fig. \ref{fig:koga} depicts the results of comparison between $K$-OGA and the benchmark algorithm in terms of the average number of selected sensors for two values of $K$. In the simulations, $W=100$ m, $R_s=10$ m  and $\alpha=45 \degree$.
As seen, $K$-OGA selects less sensors compared to the benchmark. The difference is greater when the number of available sensors is less or when $K$ is greater. 

The benchmark algorithm is essentially equivalent to forming $K$ independent weak barrier without taking into account that some of the targets from the $i^{th}$ step are already $i+1$-covered, i.e. repeated coverage. It results in higher cost (i.e. more selected sensors) in the benchmark algorithm compared to the optimal greedy $K$-OGA.

\begin{figure}[!t]
\centering
\includegraphics[height=0.3\textheight, width=0.5\textwidth]{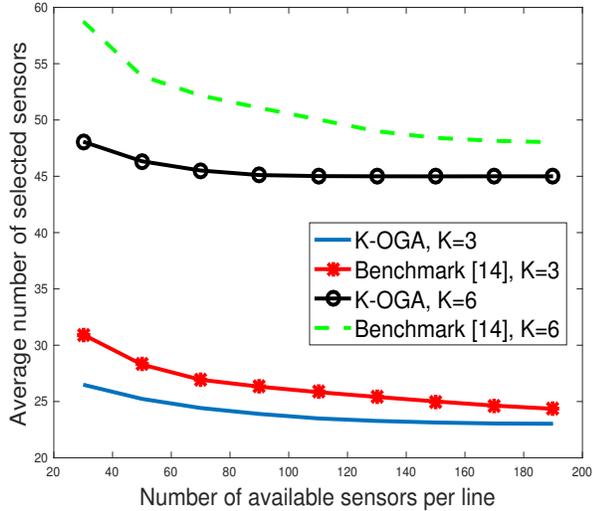}
\caption{The performance comparison of $K$-OGA and the benchmark for two different values of $K$}
\label{fig:koga}
\end{figure}

%%%%%%%%%%%%%%%%%%%%%%%%%%%%%%%%%
\subsection{Gap mending for single sensor failure}

In this scenario, we study the performance of LOGM in case of single sensor failure. We assume that OGA has been previously applied to select an optimal set of sensors. Then, a single sensor fails and we compare the performance of LOGM with a new round of OGA in mending the gap.

Fig. \ref{fig:single_failure} shows the results of applying LOGM and the new OGA in the problem of single sensor failure for different number of available sensors. The settings are the same as the previous coverage scenario, i.e. $W=1000$ m, $R_s=10$ m, $\alpha=45 \degree$, except that the deployment is random, i.e. Poisson so that enough sensors are available for mending the gap. For each $n$ (number of available sensors), we perform 1000 random realizations of sensors locations based on a line-based deployment. Then, we start removing the sensors selected by the old OGA method and compute the difference in the number of sensors selected by LOGM and the new OGA. Among the 1000 realizations for each  $n$, we compute the maximum and the average of the cost difference.

As seen, the maximum cost difference is $1$ over the entire interval, which confirms Theorem \ref{thm:LOGM_max_1}. 
%The non-abrupt slope is due to  for every 10 sensors, otherwise it has to be an abrupt drop from $1$ to $0$. 
Furthermore, the average difference is less than $0.2$ for the depicted range of number of available sensors ($n$). For large $n$, the maximum over all realizations become zero, which means both methods select the same number of sensors.

%%%%%%%%%%%%%%%%%%%%%
\begin{figure}[!t]
\centering
\includegraphics[height=0.3\textheight, width=0.5\textwidth]{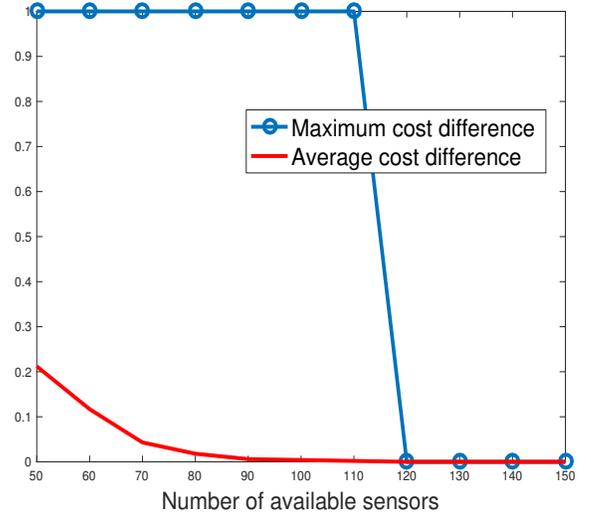}
\caption{The comparison between LOGM and the new OGA performance in the single sensor failure scenario}
\label{fig:single_failure}
\end{figure}

%%%%%%%%%%%%%%%%%%%%%%%%%%%%%
\subsection{Gap mending for multiple sensors failure (multi-gap scenario)}
 
% We used number of available sensors per line to be able to plot the results in one figure
Here, we study the performance of LOGM in the case of multiple sensors failure and compare it with that of the optimal method, i.e. a new round of OGA. 
There are $1000$ sensors in this scenario and we have $W=100$ m, $R_s=2$ m, $\alpha=45 \degree$. The sensor deployment settings are the same as the single gap scenario. 

Fig. \ref{fig:multi_gap}, depicts the maximum and average difference in the number of selected sensors in both methods. As seen, both average and maximum differences are increasing linearly with the number of gaps. 
In addition, the maximum values are in concord with theorem \ref{thm:multi_gap}, i.e. equal to $2m-1$, where $m$ is the number of gaps.
%, i.e. at most $m$ more sensors selected by LOGM when there are $m$ gaps.
These observations demonstrate the acceptable scalability of LOGM in the number of gaps.

%%%%%%%%%%%%%%%%%%%%%%%%%%%%%%%%%
\section{Conclusion}
\label{sec:conclusion}
The sensor selection for line coverage problem is redefined and an optimal greedy algorithm, called OGA, is proposed to solve it. 
It is shown that the existing order can be utilized to design optimal greedy algorithm with lower complexity compared to the state-of-the-art.
The better performance of OGA compared to a greedy benchmark is studied in simulations. 
Furthermore, the weak $K$-barrier coverage is modeled as an instance of the line coverage problem and solved optimally using a novel algorithm called $K$-OGA. $K$-OGA involves less computational complexity compared to the state-of-the-art. Moreover, an efficient and locally implementable algorithm for gap mending is proposed, called LOGM, which
% is analytically proved and numerically verified to 
selects at most one sensor more than the optimal. It is also demonstrated numerically that in the multi-gap scenario, the cost distance of LOGM from the optimal increases linearly with the number of gaps. 

%%%%%%%%%%%%%%%%%%%%%%%%%%%

\begin{figure}[!t]
\centering
\includegraphics[height=0.3\textheight, width=0.5\textwidth]{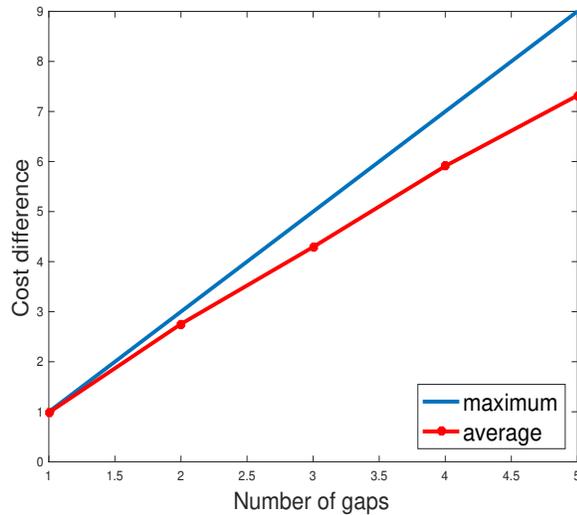}
\caption{The comparison between LOGM and the new OGA performance in the multi-gap scenario}
\label{fig:multi_gap}
\end{figure}

%%%%%%%%%%%%%%%%%%%%%%%%%%%%%%%
\bibliographystyle{IEEEtran}
\bibliography{IEEEabrv,Journal}

%%%%%%%%%%%%%%%%%%%%%%%

\ifCLASSOPTIONcaptionsoff
  \newpage
\fi
%%%%%%%%%%%%%%%%%%%%%%%%%%
\ignore{
\begin{IEEEbiography}{Michael Shell}
Biography text here.
\end{IEEEbiography}

\begin{IEEEbiographynophoto}{John Doe}
Biography text here.
\end{IEEEbiographynophoto}

%\newpage

\begin{IEEEbiographynophoto}{Jane Doe}
Biography text here.
\end{IEEEbiographynophoto}
}

\end{document}